\documentclass[11pt,a4paper]{article}

\usepackage{indentfirst,mathrsfs}
\usepackage{amsfonts,amsmath,amssymb,amsthm}
\usepackage{latexsym,amscd}
\usepackage{amsbsy}
\usepackage{color}
\usepackage[noadjust]{cite}
\usepackage{cite}
\usepackage{multirow}
\usepackage{makecell}
\usepackage{float}
\usepackage{cases}
\newtheorem{thm}{Theorem}
\newtheorem{lem}{Lemma}
\newtheorem{cor}{Corollary}

\newtheorem{defn}{Definition}
\newtheorem{example}{Example}
\newtheorem{remark}{Remark}

\newcommand{\F}{{\mathbb F}}

\newcommand{\fnn}{{\mathbb F}_{2^{n}}}
\newcommand{\fn}{{\mathbb F}_{2^{2m}}}

\newcommand{\bm}{\overline{b}}

\begin{document}

\title{Complete characterization of the differential spectrum of a Niho type power function}
\author{Nian Li, Xi Xie, Rui Xu, Yi Yu, Xiangyong Zeng
\thanks{N. Li, X. Xie and Y. Yu are with the Hubei Provincial Engineering Research Center of Intelligent Connected Vehicle Network Security, and School of Cyber Science and Technology, Hubei University, Wuhan 430062, China.
R. Xu and X. Zeng are with the Hubei Key Laboratory of Applied Mathematics, and Faculty of Mathematics and Statistics, Hubei University, Wuhan 430062, China.
Email: nian.li@hubu.edu.cn, xi.xie@aliyun.com, 2937519468@qq.com, 275768902@qq.com, xiangyongzeng@aliyun.com.}
}
\date{}%\today
\maketitle
\begin{quote}
{{\bf Abstract:}
Power functions with Niho exponents have attracted considerable attention due to their important applications in sequence design, coding theory, and cryptography. This paper investigates the differential properties of Niho type power functions of the form $F(x)=x^{s(2^m-1)+1}$ over $\mathbb{F}_{2^{2m}}$ with $2\leq s\leq 2^m$.
We first establish a general characterization of the differential spectrum of $F(x)$ having at most three nonzero values via its Walsh spectrum.
Focusing subsequently on the case $s=(2^k+1)^{-1} \pmod{2^m+1}$ where $\gcd(k,m)=e$, we employ a refined analysis of the number of solutions to certain equations over finite fields. Specifically, it is proved that $F(x)$ is locally differentially $2^e$-uniform when $\gcd(2^k-1,2^m+1)=2^e+1$ and locally differentially $(2^{2e}-2^e)$-uniform when $\gcd(2^k-1,2^m+1)=1$, and their differential spectra are completely determined. These results completely characterize the differential properties of this family and yield new infinite families of locally differentially $4$-uniform power functions.
}

\end{quote}

\section{Introduction}
Substitution boxes (S-boxes), typically represented by vectorial Boolean functions over finite fields, are the primary source of nonlinearity in modern block ciphers. Consequently, the cryptographic properties of vectorial Boolean functions have been extensively studied. Differential cryptanalysis, introduced by Biham and Shamir \cite{BS}, is one of the most powerful techniques for evaluating the security of symmetric cryptosystems. For a mapping $F:\F_{2^n}\rightarrow\F_{2^n}$, its differential behavior is characterized by the difference distribution table (DDT), whose entries are defined by
\[{\rm DDT}_F(a,b)=\#\{x\in\F_{2^n}:F(x+a)+F(x)=b\},\]
where $a,b\in\F_{2^n}$. The differential uniformity of $F$, introduced by Nyberg \cite{Nyberg-93,SM-NK}, is defined as
\[\Delta_F=\max_{a\in\F_{2^n}^*,\,b\in\F_{2^n}}{\rm DDT}_F(a,b).\]
In addition to differential uniformity, the differential spectrum of $F$, defined as the multi-set 
$\{{\rm DDT}_{F}(a,\,b): a\in\fnn^*,\,b\in\fnn\},$
provides a more detailed description of the differential behavior of $F$ and plays an important role in the analysis of various variants of differential cryptanalysis \cite{BCC}. Therefore, the differential properties of vectorial Boolean functions have attracted considerable research attention, and for comprehensive surveys and recent advances, the reader is referred to \cite{Berger,BCC,BCC1,BDMW,Budaghyan,EM,XY} and Chapter 11 in the book \cite{Carlet-book-2021}.

Power functions constitute an important class of S-box candidates owing to their simple algebraic structure and efficient implementation. Let $F(x)=x^d$ be a power function over $\F_{2^n}$. It is well known that ${\rm DDT}_F(a,\,b)={\rm DDT}_F\left(1,\,b/a^d\right)$ for every $a\in\F_{2^n}^*$ and $b\in\F_{2^n}$. Hence, the differential spectrum of $F$ is completely determined by the values of ${\rm DDT}_F(1,b)$, making power functions particularly amenable to differential analysis. The values ${\rm DDT}_F(1,0)$ and ${\rm DDT}_F(1,1)$ are often exceptional and may obscure the intrinsic differential behavior of $F$. Motivated by this observation, Blondeau, Canteaut and Charpin \cite{BCC1} introduced the notion of locally-APN power functions, which was later generalized by Blondeau and Perrin \cite{BP} to locally differentially $\delta$-uniform power functions. A power function $F$ over $\F_{2^n}$ is said to be locally differentially $\delta$-uniform if $\max\{{\rm DDT}_F(1,b):b\in\F_{2^n}\setminus\F_2\}=\delta$. In particular, when $\delta=2$, such functions are
referred to as locally almost perfect nonlinear (locally-APN) power functions, which offer superior resistance against differential cryptanalysis compared to other functions of equivalent differential uniformity. Until recently, only a few power functions over $\F_{2^n}$ with low locally differential uniformity and known differential spectra have been studied, as summarized in Table \ref{differential-table}. This motivates further investigations into such functions and their differential spectra.

\begin{table}[!htb]\footnotesize
\caption{The power function $F(x)=x^d$ over $\fnn$ with known differential spectrum}  \label{differential-table}
\renewcommand\arraystretch{0.8}
\setlength\tabcolsep{2pt}
\centering
\begin{tabular}{lllll}
\hline Exponent $d$                & Condition                   &  $\Delta(F)$ &  $\delta$   & Refs.                  \\ \hline
    $2^t+1$            &    $\gcd(t,n)=s$            &    $2^s$   & $2^s$    &\cite{BCC,EM}            \\ 
    $2^{2t}-2^t+1$     & $\gcd(t,n)=s$, $n/s$ odd    &    $2^s$  &$2^s$     &\cite{BCC}            \\
    $2^n-2$            &$n\geq 2$                    & $2$ or $4$  & $2$  &\cite{BCC,EM}            \\ 
   $2^{2k}+2^k+1$      &$n=4k$                       & 4    & $2$ or $4$          &\cite{BCC,XY}       \\ 
   $2^t-1$             &$t=3,n-2$                    & 6 or 8    & $6$      &\cite{BCC1}             \\ 
  $2^t-1$             &$n=2m$, $t=m, m+1$    & $2^{m}-2$ or $2^{m}$&  $2$   &\cite{BCC1}   \\ 
   $2^t-1$             &$t=(n-1)/2,(n+3)/2$, $n$ odd & $6$ or $8$   &$6$  &\cite{BP}             \\ 
   $2^{3k}+2^{2k}+2^{k}-1$ & $n=4k$                  & $2^{2k}$  &$2^{2k}-2^k$     &\cite{LWZT,Kim2022,Tu2023}             \\ 
   $2^m+2^{(m+1)/2}+1$ &$n=2m$, $m\geq5$ odd         & $8$        &$8$    &\cite{XYY}             \\ 
   $2^{m+1}+3$         &$n=2m$, $m\geq5$ odd         & $8$           &$8$  &\cite{XYY}             \\
   $k(2^m-1)$         & $n=2m$, $\gcd(k,2^m+1)=1$        &$2^m-2$ & $2$             &   \cite{HLXZT} \\
   $(2^m\!-\!1)(2^k\!+\!1)^{-1}\!+\!1$        & $n=2m$, $\gcd(k,m)=1$  &$2^m$ & $2$   & \cite{Xie-locally} \\
   $2^m+3$         & $n=2m$      &$2^m$ or $2^m+2$ & $4$ or $6$            &   \cite{Li-Yan23} \\
   $3\cdot 2^m-2$         & $n=2m$, $m\geq 4$ even       &$2^m$ & $4 $            &   \cite{YanLi2026} \\
   $k(2^m-1)$         & $n=2m$, $\gcd(k,2^m+1)=t\geq 1$        &$t(2^m\!-\!1)\!-\!1$ & $2t^2$            &   \cite{CL2025} \\
   $(2^m\!-\!1)(2^k\!+\!1)^{-1}\!+\!1$        & $n=2m$, $\gcd(k,m)\!=\!e$, $\gcd(2^k\!\!-\!\!1,2^m\!\!+\!\!1)\!=\!2^e\!+\!1$  &$2^m$ & $2^e$ & This paper \\
   $(2^m\!-\!1)(2^k\!+\!1)^{-1}\!+\!1$        & $n=2m$, $\gcd(k,m)\!=\!e$, $\gcd(2^k\!-\!1,2^m\!+\!1)=1$  &$2^m$ & $2^{2e}-2^e$ & This paper \\\hline
\end{tabular}\\
-Here $\Delta(F)$ and $\delta$ denote the differential uniformity and locally differential uniformity of $F$ respectively. 
\end{table}

Niho type power functions are of particular interest in view of their rich connections with sequence design, coding theory, and cryptography. 
In this paper, we study a class of Niho type power functions over $\F_{2^{2m}}$ of the form
\[F(x)=x^{s(2^m-1)+1},\quad 2\le s\le 2^m.\]
It has been shown that every Niho type power function satisfies
${\rm DDT}_F(1,1)\geq 2^m$, indicating that the exceptional differential behavior is concentrated at the point $b=1$ \cite{Ranto-Rosendahl06}. We begin by characterizing the differential spectrum of $F$ under the assumption that the values ${\rm DDT}_F(1,b)$, for $b\neq1$, take at most three distinct values, by means of the Walsh transform of $F$. We then turn our attention to the special case $s=\left(2^k+1\right)^{-1}$, where $\left(2^k+1\right)^{-1}$ denotes the multiplicative inverse of $2^k+1$ modulo $2^m+1$. 
It was proved in \cite{Xie-locally} that this family is locally-APN when $\gcd(k,m)=1$, and conjectured that it may comprise all Niho type locally-APN power functions over $\F_{2^n}$. Motivated by this conjecture, we further investigate the case  $\gcd(k,m)>1$. Let $\gcd(k,m)=e$. For $e>1$, the associated differential equations exhibit a more intricate solution structure, requiring a refined analysis of their solution sets.
By analyzing the relation between the roots of two equations of algebraic degree two, we determine the exact numbers of solutions to the differential equations according to 
the value of $\gcd\left(2^k-1,2^m+1\right)$.
Consequently, we prove that 
$F$ is locally differentially $2^e$-uniform when $\gcd\left(2^k-1,2^m+1\right)=2^e+1$, and locally differentially $\left(2^{2e}-2^e\right)$-uniform when $\gcd\left(2^k-1,2^m+1\right)=1$.
Furthermore, by employing the Walsh spectrum, we completely characterize the differential spectrum of $F$. Extensive computational evidence suggests that our results include all locally differentially $4$-uniform Niho type power functions beyond the previously known families.

The remainder of this paper is organized as follows. Section \ref{prel} presents the necessary preliminaries. Section \ref{diff-four} establishes general results on the differential spectrum of Niho type power functions. Section \ref{diff1} focuses on the family $F(x)=x^{s(2^m-1)+1}$ with $s=\left(2^k+1\right)^{-1}$. Finally, Section \ref{conc} concludes the paper.

\section{Preliminaries}\label{prel}
Let \(q=2^m\). Throughout this paper, we work over the finite field
\(\F_{q^2}\), and denote by \(\overline{x}=x^q\) the conjugate of \(x\) over
\(\F_q\). For a finite set \(E\), let \(\#E\) denote its cardinality. We also define the unit circle
\[\mu_{q+1}=\left\{x\in\F_{q^2}: x\overline{x}=1\right\}.\]
The absolute trace from \(\F_{2^n}\) to \(\F_2\) is given by
    $\operatorname{Tr}_1^n(z)=\sum_{i=0}^{n-1}z^{2^i}.$
For a function $F(x)$ mapping from $\F_{2^n}$ to $\F_{2^n}$, define the exponential sum
\[W_F(a,b)=\sum_{x\in\F_{2^n}}(-1)^{{\rm Tr}_1^{n}\left(aF(x)+bx\right)}\]
for $a,\,b\in\F_{2^n}$. In particular, when $a=1$, the sum
$W_F(b)=\sum_{x\in\F_{2^n}}(-1)^{{\rm Tr}_1^n(F(x)+bx)}$ is called the Walsh transform of $F(x)$ at the point $b\in\F_{2^n}$. 

We first recall several auxiliary results that will be used throughout
the paper.

\begin{lem}\label{lem:gcd}
    Let $m,n$ be positive integers, and let $e=\gcd(m,n)$. Then
    \[
    \gcd\left(2^m-1,\;2^n-1\right)=2^{e}-1.
    \]
    Moreover,
        \[
    \gcd\left(2^m+1,\;2^n-1\right)=
    \begin{cases}
        1, &\text{if } n/e \text{ is odd},\\[6pt]
        2^{e}+1, &\text{if } n/e \text{ is even},
    \end{cases}
    \]
    and
        \[
    \gcd\left(2^m+1,\;2^n+1\right)=
    \begin{cases}
        2^{e}+1, &\text{if both } m/e \text{ and } n/e \text{ are odd},\\[6pt]
        1, &\text{if one of } m/e,\,n/e \text{ is even}.\\[6pt]
    \end{cases}
    \]
\end{lem}

% \begin{lem}[Lemma 5, \cite{Xie-DLCT}]\label{lem.x}
% 	Let $n=2m$ be a positive integer. Then each element $x\in\F_{2^n}\backslash\F_{2^m}$ can be uniquely written as $x=v_1(v_2+1)/(v_1+v_2)$, where $v_1\ne v_2\in\mu_{2^m+1}\backslash \{1\}$.
% \end{lem}

\begin{lem}[Theorem 6, \cite{MKCLG}]\label{lem.eq-tri}
Let $q=2^m$, where $m$ and $r$ are positive integers.
Then the trinomial $x^q+x+b$ has either $q$ or no roots in $\mathbb{F}_{q^r}$,
according as ${\rm Tr}_m^{mr}(b)=0$ or ${\rm Tr}_m^{mr}(b)\neq 0$.
\end{lem}

The following lemma, which follows from Lemma 22 and the proof of
Theorem 23 in \cite{DFHR}, will also be used in proving our main result.

\begin{lem}[Lemma 22, \cite{DFHR}]\label{lem.eq}
Let $n$ and $r$ be positive integers with $r_0=\gcd(r,\,n)$. Then the polynomial
$$Q(x)=x^{2^r+1}+ax^{2^r}+bx+c\in\mathbb{F}_{2^n}[x]$$
has either $0,\,1,\,2$ or $2^{r_0}+1$ roots in $\mathbb{F}_{2^n}$.
Specially, if $n=2m$ and $x_0,\,x_1,\,x_2\in\mu_{2^m+1}$ are three distinct roots of $Q$, then
$Q(x)$ has $2^{r_1}+1$ distinct roots in $\mu_{2^m+1}$, where $r_1=\gcd(r_0,\,m)$.
Further, $x_0+Ax_1+(1+A)x_2\ne0$ for any $A\in\mathbb{F}_{2^{r_1}}$ and each root of $Q(x)$ in $\mu_{2^m+1}\backslash\{x_0,\,x_1,\,x_2\}$ can be parameterized as
$$x_A=\frac{x_1x_2+Ax_0x_2+(1+A)x_0x_1}{x_0+Ax_1+(1+A)x_2} \qquad {\rm for }\,\,A\in\F_{2^{r_1}}\setminus\F_2.$$
\end{lem}

\section{Four-valued differential spectra of Niho type power functions}
\label{diff-four}\label{diff-four}

In this section, we characterize the differential spectra of Niho type power functions for which \({\rm DDT}_F(1,b)\), \(b\neq1\), 
take at most three distinct values. 

We first recall the definition of Niho exponents. A positive integer $d$ is called a Niho exponent with respect to $\F_{q^2}$ if $d\equiv 2^i\,({\rm mod}\, q-1)$ for some nonnegative integer $i$. The case \(i=0\) is referred to as a normalized Niho exponent.
For Niho type power functions, \({\rm DDT}_F(1,1)\) is an exceptional value determined explicitly in \cite{Ranto-Rosendahl06}.

\begin{lem}[Theorem 7, \cite{Ranto-Rosendahl06}]\label{lem.eq-b=1}
Let \(F(x)=x^d\) be a power function over \(\F_{2^{2m}}\), where $d=s(2^m-1)+1$ and \(2\leq s\leq 2^m\).
Define $r_0=\gcd(s,\,2^m+1)$ and $r_1=\gcd(s-1,\,2^m+1)$. Then 
$${\rm DDT}_F(1,1)=2^m+(r_0-1)(r_0-2)+(r_1-1)(r_1-2).$$
\end{lem}

Therefore, it remains to determine the values
\({\rm DDT}_F(1,b)\) for \(b\neq1\).
We first introduce the differential spectrum of a power function.

\begin{defn}
Let \( F(x) = x^d \) be a power function over \( \mathbb{F}_{2^n} \) with differential uniformity \(\Delta_F\), where $n$ and $d$ are positive integers. Define
\[\omega_i = \# \{b \in \mathbb{F}_{2^n} : {\rm DDT}_F(1,b)= i \}, \quad 0 \leq i \leq \Delta_F.\]
The differential spectrum of \( F \) is defined as the multiset  
\[\mathbb{DS}_F = \{\omega_i > 0 : 0 \leq i \leq \Delta_F\}.\]
\end{defn}

There are basic transformations which preserve the differential spectrum.

\begin{lem}[Lemma 1, \cite{BCC}]\label{lem-eqv}
Let \( F_d(x) = x^d \) and \( F_t(x) = x^t \) over $\F_{2^n}$. If there exists \( k \) such that \( t \equiv 2^k d \pmod{2^n - 1} \), or if \( \gcd(2^n - 1, d) = 1 \) and \( t \equiv d^{-1} \pmod{2^n - 1} \), then \( F_d \) and \( F_t \) have the same differential spectrum.
\end{lem}

We begin with a general result for power functions, which will later be specialized to Niho type power functions.

\begin{thm}\label{thm.diff-Wal-four}
Let \(F(x)=x^d\) be a power function over \(\F_{2^n}\).
Suppose that $\mathbb{DS}_F=\{\omega_0,\,\omega_{t_1},\omega_{t_2},\,\omega_{N}\}$,
where $t_1,t_2,N$ are positive integers with \(t_1<t_2\) and $\omega_{N}=1$. Then
\begin{eqnarray*}
&&\omega_0=2^n-1-\frac{N^2-\aleph_F+(t_2+t_1)(2^n-N)}{t_1t_2};\\
&&\omega_{t_1}=\frac{N^2-\aleph_F+(2^n-N)t_2}{t_1(t_2-t_1)};\\
&&\omega_{t_2}=\frac{\aleph_F-N^2+t_1(N-2^n)}{t_2(t_2-t_1)}.
\end{eqnarray*}
Here $\aleph_F=\frac{\sum_{a,b\in\F_{2^n}\backslash\{(0,\,0)\}}W_F(a,\,b)^4}{2^{2n}(2^n-1)}$. Particularly, $\aleph_F=2^{-2n}\sum_{b\in\F_{2^n}^*}W_F(b)^4$ if $\gcd(d,\,2^n-1)=1$.
\end{thm}

\begin{proof}
According to the definition of $W_{F}(a,b)$, we have
\begin{equation}\label{eq.W-N4} 
\sum_{a,b\in\F_{2^n}}\!\!W_{F}(a,b)^4=\sum_{a,b\in\F_{2^n}}\sum_{x_1,x_2,x_3,x_4\in\F_{2^n}}
(-1)^{{\rm Tr}_1^n\left(a(x_1^d+x_2^d+x_3^d+x_4^d)+b(x_1+x_2+x_3+x_4)\right)}
= 2^{2n}N_4    
\end{equation}
% \begin{eqnarray}
% \sum_{a,b\in\F_{2^n}}\!\!W_{F}(a,b)^4
% &&=\sum_{a,b\in\F_{2^n}}\sum_{x_1,x_2,x_3,x_4\in\F_{2^n}}
% (-1)^{{\rm Tr}_1^n\left(a(x_1^d+x_2^d+x_3^d+x_4^d)+b(x_1+x_2+x_3+x_4)\right)} \nonumber \\ 
% &&= 2^{2n}N_4 \label{eq.W-N4} 
% \end{eqnarray}
and $N_4$ is the number in $(\F_{2^n})^4$ of
\[\left\{
\begin{array}{l}
x_1+x_2+x_3+x_4=0,\\[5pt]
x_1^d+x_2^d+x_3^d+x_4^d=0,
\end{array}
\right.\]
which is equivalent to
\begin{equation}\label{eq.N4}
\left\{
\begin{array}{l}
x_1+x_2=x_3+x_4=u,\\[5pt]
x_1^d+(x_1+u)^d=x_3^d+(x_3+u)^d=v.
\end{array}
\right.
\end{equation}
When $u=0$, then \eqref{eq.N4} has solutions if and only if $v=0$ and thus has $2^{2n}$ solutions. When $u\ne 0$, dividing both sides of the equations by $u^d$ and letting $y_1=x_1/u$, $y_2=x_3/u$, we deduce
\[\left\{
\begin{array}{l}
x_2=u(y_1+1), x_4=u(y_2+1),\\[5pt]
y_1^d+(y_1+1)^d=y_2^d+(y_2+1)^d=v/u^d.
\end{array}
\right.\]
This implies that $N_4=2^{2n}+(2^n-1)\left(t_1^2\omega_{t_1}+t_2^2\omega_{t_2}+N^2\omega_{N}\right)$. Combining this with $\omega_{N}=1$ and $W_F(0,\,0)=2^n$, \eqref{eq.W-N4} gives
\[\sum_{a,b\in\F_{2^n}\backslash\{(0,\,0)\}}\!\!
W_F(a,\,b)^4=2^{2n}(2^n-1)\left(t_1^2\omega_{t_1}+t_2^2\omega_{t_2}+{N}^2\right).\]
Denote $\aleph_F=\frac{\sum_{a,b\in\F_{2^n}\backslash\{(0,\,0)\}}W_F(a,\,b)^4}{2^{2n}(2^n-1)}$. Together with the standard identities
\begin{eqnarray*}
&&\omega_0+\omega_{t_1}+\omega_{t_2}
+\omega_{N}=2^{n}; \\
&&t_1\omega_{t_1}+t_2\omega_{t_2}
+N\omega_{N}=2^{n},
\end{eqnarray*}
a calculation gives the values of $\omega_0$, $\omega_{t_1}$ and $\omega_{t_2}$ as desired.
Particularly, if $\gcd(d,\,2^n-1)=1$, then for any $a\in\F_{2^n}^*$, there exists unique $c\in\F_{2^n}^*$ such that $ac^d=1$ and then $W_F(a,\,b)=W_F(1,\,bc)=W_F(bc)$. Combining with $W_F(a,\,0)=0$ for $a\in\F_{2^n}^*$, we have
$\aleph_F=\frac{(2^n-1)\sum_{b\in\F_{2^n}^*}W_F(b)^4}{2^{2n}(2^n-1)}
=2^{-2n}\sum_{b\in\F_{2^n}^*}W_F(b)^4$.
This completes the proof.
\end{proof}

We next focus on Niho type power functions. 
To investigate their differential properties, we first recall the following characterization of the exponential sum \(W_F(a,b)\).

\begin{lem}[{Lemma 2, \cite{Li-HKT}}]\label{lem.WF}
Let $F(x)=x^d$ with $d=s(2^m-1)+1$, where $2\leq s\leq 2^m$ is an integer. Then for any $a,b\in\F_{2^{2m}}$, the exponential sum
\[W_F(a,b)=(N(a,b)-1)2^m,\]
where $N(a,b)$ denotes the number of elements $u\in \mu_{2^m+1}$ satisfying
\[au^s+\overline{a} u^{1-s}+\overline{b}u+b=0.\]
\end{lem}

When \(a=1\), the exponential sum \(W_F(a,b)\) reduces to the Walsh transform of \(F\), whose possible values in the four-valued case were determined in \cite{Ranto-R}. 
Using this result, we derive the value distribution of \(W_F(a,b)\).

% \begin{lem}[Theorem 12, \cite{Ranto-R}]\label{lem.WF}
% Let $F(x)=x^d$ and $d$ is of Niho type. If the Walsh transform of $F(x)$ is four-valued, then $W_F(b)\in\{-2^m, 0, 2^m, 2^{m+j} \},$
% where $b\in\F_{2^n}$ and $j>0$ satisfies $2j\mid m$.
% \end{lem}

\begin{lem}\label{lem-Wal4}
Let $n=2m$ and $F(x)=x^d$ with $d=s(2^m-1)+1$, where $2\leq s\leq 2^m$. Denote $r_0=\gcd(s,\,2^m+1)$ and $r_1=\gcd(s-1,\,2^m+1)$. If $W_F(a,\,b)$ takes at most four values as $a,b$ run over $\F_{2^{n}}\times\F_{2^{n}}\backslash\{(0,\,0)\}$, then the value distribution of $W_F(a,\,b)$ is given by
\begin{equation*}\left\{\begin{array}{rll}
			-2^m,&\,\, 2^{m-1}(2^n-1)(2^m-1-\frac{N-1}{2^j+1})\,&\,{\rm times},\\
			0,&\,\,(2^n-1)(2^{m-j}(N-1)+1)\,&\,{\rm times},\\
			2^m,&\,\,2^{m-1}(2^n-1)(2^m+1-\frac{N-1}{2^j-1})\,&\,{\rm times},\\
			2^{m+j},&\,\,\frac{2^{m-j}(2^n-1)(N-1)}{2^{2j}-1}\,&\,{\rm times},\\
		\end{array}\right.
	\end{equation*}
    where $N=2^m+(r_0-1)(r_0-2)+(r_1-1)(r_1-2)$ and $j>0$ with $2j|m$.
\end{lem}

\begin{proof}
According to the definition of $W_{F}(a,b)$, for $r\geq1$, we have
\begin{equation}\label{eq-WF}
\sum_{a,b\in\F_{2^n}}\!\!W_{F}(a,b)^r=\sum_{a,b\in\F_{2^n}}\sum_{x_1,x_2,\dots,x_r\in\F_{2^n}}
(-1)^{{\rm Tr}_1^n\left(a(x_1^d+x_2^d+\cdots+x_r^d)+b(x_1+x_2+\cdots+x_r)\right)}
= 2^{2n} N_r    
\end{equation}
and $N_r$ by the number of solutions in $(\F_{2^n})^r$ of
\begin{equation}\label{x1-xr}
\left\{
\begin{array}{l}
x_1+x_2+\cdots+x_r=0,\\[5pt]
x_1^d+x_2^d+\cdots+x_r^d=0.
\end{array}
\right.
\end{equation}
If $W_F(a,\,b)$ takes at most four values as $a,b$ run over $\F_{2^{n}}\times\F_{2^{n}}\backslash\{(0,\,0)\}$, from \cite[Theorem 12]{Ranto-R} one knows
$$W_F(a,b)\in\left\{-2^m,\, 0,\, 2^m,\, 2^{m+j}\right\},$$
where $j>0$ satisfies $2j\mid m$. Let $M_i$ denote the pairs of $(a,b)\in\F_{2^n}\times\F_{2^n}$ such that $W_F(a,b)=i$. Combining with $W_F(0,\,0)=2^n$, \eqref{eq-WF} yields
\begin{equation}\label{eq-Mi}
\begin{aligned}
&M_{-2^m}+M_0+M_{2^m}+M_{2^{m+j}}+1=2^{2n}; \\
&-2^mM_{-2^m}+2^mM_{2^m}+2^{m+j}M_{2^{m+j}}+2^n=2^{2n}N_1;\\
&2^nM_{-2^{m}}+2^nM_{2^m}+2^{n+2j}M_{2^{m+j}}+2^{2n}=2^{2n}N_2;\\
&-2^{3m}M_{-2^m}+2^{3m}M_{2^m}+2^{3m+3j}M_{2^{m+j}}+2^{3n}=2^{2n} N_3.
\end{aligned}
\end{equation}
According to the definition of $N_r$, it can be readily derived that $N_1=1$ and $N_2=2^n$. When $r=3$, if $x_3=0$, equation \eqref{x1-xr} turns into the case for $r=2$; if $x_3\ne0$, dividing both sides of equations in \eqref{x1-xr}, then \eqref{x1-xr} becomes
\[\left\{
\begin{array}{l}
y_1+y_2=1,\\[5pt]
y_1^d+y_2^d=1,
\end{array}
\right.\]
where $y_1=\frac{x_1}{x_3}$ and $y_2=\frac{x_2}{x_3}$. Combining with Lemma \ref{lem.eq-b=1} we obtain that $N_3=2^n+(2^n-1)N$ and $N=2^m+(r_0-1)(r_0-2)+(r_1-1)(r_1-2)$. Substituting the values of $N_1$, $N_2$ and $N_3$ into \eqref{eq-Mi},  we obtain the desired result by solving this system of equations. This completes the proof.
\end{proof}

By Lemma \ref{lem-Wal4}, if \(W_F(a,b)\) takes exactly four distinct values, then $\aleph_F=2^n+2^{m+j}(N-1).$
Substituting this expression into Theorem \ref{thm.diff-Wal-four} immediately yields the following corollary.

\begin{cor}\label{cor.diff-Wal}
Let $n=2m$, and let $F(x)=x^d$ be a power function over $\F_{2^n}$ with
$d=s(2^m-1)+1$ for some integer $2\le s\le 2^m$. Define $r_0=\gcd(s,\,2^m+1)$,
$r_1=\gcd(s-1,\,2^m+1)$, and $N=2^m+(r_0-1)(r_0-2)+(r_1-1)(r_1-2)$.
Suppose that ${\rm DDT}_F(1,b)\in\{0,t_1,t_2\}$ for all $b\in\F_{2^n}\setminus\{1\}$, where $0<t_1<t_2$.
If the exponential sum \(W_F(a,b)\) takes exactly four distinct values as \((a,b)\) ranges over $\F_{2^{n}}\times\F_{2^{n}}\backslash\{(0,\,0)\}$, then
$\mathbb{DS}_F=\{\omega_0,\,\omega_{t_1},\omega_{t_2},\,\omega_{N}=1\}$ with
\begin{eqnarray*}
&&\omega_0=2^n-1-\frac{N^2-2^n-2^{m+j}(N-1)+(t_2+t_1)(2^n-N)}{t_1t_2};\\
&&\omega_{t_1}=\frac{N^2-2^n-2^{m+j}(N-1)+(2^n-N)t_2}{t_1(t_2-t_1)};\\
&&\omega_{t_2}=\frac{2^n+2^{m+j}(N-1)-N^2+t_1(N-2^n)}{t_2(t_2-t_1)} \,\,\,{\rm with} \,\,\,2j|m.
\end{eqnarray*}
\end{cor}

% \begin{proof}
% If the exponential sum \(W_F(a,b)\) takes exactly four distinct values as \((a,b)\) ranges over
% $\F_{2^{n}}\times\F_{2^{n}}\backslash\{(0,\,0)\}$, then the values of \(W_F(a,b)\) are characterized by Lemma \ref{lem-Wal4}. By calculation we have $$\sum_{a,b\in\F_{2^n}\backslash\{(0,\,0)\}}W_F(a,\,b)^4=2^{2n}(2^n-1)(2^n+2^{m+j}(N-1))$$
% and then
% \[\aleph=2^n+2^{m+j}(N-1).\]
% Substituting it into Theorem \ref{thm.diff-Wal-four} gives the desired values of $\omega_0$, $\omega_{t_1}$ and $\omega_{t_2}$. This completes the proof.
% \end{proof}

\section{The differential spectrum of the Niho power function $F(x)=x^{\left(2^k+1\right)^{-1}\left(2^m-1\right)+1}$ }\label{diff1}
In this section, we determine the differential spectrum of the Niho type
power function $F(x)=x^{s(2^m-1)+1}$, where
$s=\left(2^k+1\right)^{-1}$ denotes the multiplicative inverse of $2^k+1$ modulo $2^m+1$, and $\gcd\left(2^k+1,2^m+1\right)=1$.

\begin{thm}\label{thm.diff-4}
Let $F(x)=x^{s(2^m-1)+1}$ be a power function over $\F_{2^n}$, where $n=2m$, $s=\left(2^k+1\right)^{-1}$, $m$ and $k$ are positive integers with $\gcd\left(2^k+1,\,2^m+1\right)=1$ and $\gcd(k,\,m)=e$. Suppose that ${\rm DDT}_F(1,b)\in\{0,t_1,t_2\}$ for $b\in\F_{2^n}\backslash\{1\}$, where
$0<t_1<t_2$. Then
$\mathbb{DS}_F=\{\omega_0,\,\omega_{t_1},\omega_{t_2},\,\omega_{2^m}=1\}$ with
\[\omega_0=2^n-1-\frac{2^m\left(t_2+t_1-2^{e}\right)(2^m-1)}{t_1t_2},\,\omega_{t_1}=\frac{2^m(2^m-1)(t_2-2^e)}{t_1(t_2-t_1)},\,\omega_{t_2}=\frac{2^m(2^m-1)(2^e-t_1)}{t_2(t_2-t_1)}.\]
% \begin{eqnarray*}
% % \nonumber to remove numbering (before each equation)
% &&\omega_0=2^n-1-\frac{2^m(t_2+t_1-2^{e})(2^m-1)}{t_1t_2};\\
% &&\omega_{t_1}=\frac{2^m(2^m-1)(t_2-2^e)}{t_1(t_2-t_1)};\\
% &&\omega_{t_2}=\frac{2^m(2^m-1)(2^e-t_1)}{t_2(t_2-t_1)}.
% \end{eqnarray*}
\end{thm}

\begin{proof}
We first determine the possible values of the exponential sum $W_F(a,b)$ for $(a,b)\in\F_{2^{n}}\times\F_{2^{n}}\backslash\{(0,\,0)\}$. By Lemma \ref{lem.WF},
\[W_F(a,b)=(N(a,b)-1)2^m,\]
where $N(a,b)$ is the number of elements $u\in\mu_{2^m+1}$ satisfying
\[a u^{s}+\overline{a}u^{1-s}+\overline{b}u+b=0.\]
Since $\gcd(2^k+1,2^m+1)=1$, replacing $u$ by $u^{2^k+1}$ yields
\begin{equation}\label{eq-DS}
\overline b u^{2^k+1}+\overline a u^{2^k}+a u+ b=0.
\end{equation}
When $b=0$, equation \eqref{eq-DS} reduces to $\overline a u^{2^k}+a u=0$, which has either \(\gcd\left(2^k-1,2^m+1\right)\) solutions or
no solution in $\mu_{2^m+1}$. By Lemma \ref{lem:gcd}, $\gcd(2^k-1,2^m+1)=1$ or $2^e+1$. 
When $b\ne 0$, Lemma \ref{lem.eq} implies that equation \eqref{eq-DS} has $0,1,2$ or $2^{e}+1$ solutions in $\mu_{2^m+1}$. Combining the two cases above, it follows that
$W_F(a,b)\in\{-2^m, 0, 2^m, 2^{e+m}\}$ for $(a,b)\in\F_{2^{n}}\times\F_{2^{n}}\backslash\{(0,\,0)\}$, corresponding to $j=e$ in Corollary \ref{cor.diff-Wal}.
It remains to determine the value of $N$ in Corollary \ref{cor.diff-Wal}.
Since
$s=\left(2^k+1\right)^{-1}$ and
$\gcd\left(2^k+1,2^m+1\right)=1$, we have
$\gcd(s,2^m+1)=\gcd(s-1,2^m+1)=1.$ Hence, according to the definition of \(N\) in Corollary \ref{cor.diff-Wal}, we have \(N=2^m\).
Substituting $j=e$ and $N=2^m$ into
Corollary \ref{cor.diff-Wal} gives the desired values of
$\omega_0$, $\omega_{t_1}$, and $\omega_{t_2}$.
\end{proof}

Theorem \ref{thm.diff-4} reduces the determination of the differential
spectrum to characterizing the possible values of \({\rm DDT}_F(1,b)\) for \(b\in\F_{2^n}\setminus\{1\}\).
To proceed further, we invoke the following lemma from \cite{Xie-locally}, which gives an explicit expression for \({\rm DDT}_F(1,b)\).

\begin{lem}[Lemma 5, \cite{Xie-locally}]\label{lem.1}
Let $F(x)=x^{s(2^m-1)+1}$ be a power function over $\fn$, where \(s=\left(2^k+1\right)^{-1}\), $m$ and $k$ are positive integers with $\gcd\left(2^k+1,\,2^m+1\right)=1$.
Then for $b\in\fn\backslash\{1\}$, ${\rm DDT}_F(1,\,b)=\#\Phi$, where
\[\Phi=\left\{y\in\Omega: \left(b+\bm^{2^k}\right)y^{2^{2k}+1}+\left(\bm^{2^k+1}+1\right) y^{2^{2k}}+ \left(b^{2^k+1}+1\right)y+b^{2^k}+\bm= 0\right\}\]
and
\[\Omega=\left\{y\in\mu_{2^m+1}: y \ne 1, y^{2^k}+b\ne0, {\rm \;and}\; y\ne \frac{\bm y^{2^k}+ 1}{y^{2^k}+b} \right\}.\]
\end{lem}

For subsequent analysis, we employ the following equivalent
characterization of \({\rm DDT}_F(1,b)\).

\begin{lem}\label{lem.b2}
Let $m$ and $k$ be positive integers with $\gcd\left(2^k+1,\,2^m+1\right)=1$. Let $b\in\fn\backslash\{1\}$ and define
\begin{eqnarray}
&&T_1(y)=\left(b+\bm^{2^k}\right)y^{2^{2k}+1}+\left(\bm^{2^k+1}+1\right) y^{2^{2k}}+ \left(b^{2^k+1}+1\right)y+b^{2^k}+\bm, \label{eq.T1}\\
  &&T_2(y)=y^{2^k+1}+\overline b y^{2^k}+by+1. \label{eq.T2}
\end{eqnarray}
For $i=1,2$, let
$N_i=\left\{y\in\mu_{2^m+1}: T_i(y)=0\right\}$.
Then ${\rm DDT}_F(1,\,b)=\#N_1-\#N_2$.
\end{lem}

\begin{proof}
By Lemma \ref{lem.1}, we have ${\rm DDT}_F(1,\,b)=\#\Phi$ for any $b\in\fn\backslash\{1\}$, where the set \(\Phi\) is defined as in Lemma \ref{lem.1}.
From the definition of $T_1(y)$ and $T_2(y)$, it follows that $\Phi$ consists of all elements $y\in\mu_{2^m+1}$ such that
\[T_1(y)=0,\,\,T_2(y)\ne0,\,\,y\ne1,\,\,y^{2^k}+b \ne 0.\]
Denote $N_i=\{y\in\mu_{2^m+1}: T_i(y)=0\}$ for $i=1,2$.
Equivalently,
\begin{equation}\label{eq-N1N2}
\Phi=N_1\setminus(C\cup N_2)=(N_1\setminus C)\cap (N_1\setminus N_2),
\end{equation}
where $C=\{1\}\cup\left\{y\in\mu_{2^m+1}: y^{2^k}+b= 0\right\}$.
A direct computation shows that
\begin{equation}\label{eq:key-identity}
T_1(y)=\left(y+\overline b\right)T_2(y)^{2^k}+\left(y^{2^k}+b\right)^{2^k}T_2(y).
\end{equation}
This identity implies that every solution of $T_2(y)=0$ in $\mu_{2^m+1}$ is also a solution of $T_1(y)=0$. Hence $N_2\subseteq N_1$.
We now consider the possible intersections of \(N_1\) and \(C\).

\textbf{Case 1:} \(N_1\cap C=\emptyset\).
In this case, $N_1\setminus C=N_1$ and \eqref{eq-N1N2} yields $\Phi=N_1\setminus N_2$.

\textbf{Case 2:} $N_1\cap C=\{1\}$. Then $y=1$ is a solution of $T_1(y)=0$. Substituting \(y=1\) into \(T_1(y)=0\) gives
\[\left(b+\bm^{2^k}\right)+\left(\bm^{2^k+1}+1\right)+ \left(b^{2^k+1}+1\right)+b^{2^k}+\bm=0,\]
which is equivalent to
\[\left(\bm+1\right)^{2^k+1}=(b+1)^{2^k+1}.\]
Since $b\ne1$, we obtain
$\left(\frac{\bm+1}{b+1}\right)^{2^k+1}=1.$
Note that $\frac{\bm+1}{b+1}\in\mu_{2^m+1}$. Together with the assumption $\gcd\left(2^k+1,\,2^m+1\right)=1$, this implies $\frac{\bm+1}{b+1}=1$, that is, $\bm=b$.
Consequently, we have
$$T_2(1)=1+\overline b+b+1=0,$$
which shows $1\in N_2$. Therefore, by \eqref{eq-N1N2},
$\Phi=N_1\setminus(\{1\}\cup N_2)=N_1\setminus N_2$.

\textbf{Case 3:} $N_1\cap C=\left\{y\in\mu_{2^m+1}: y^{2^k}+b= 0\right\}$. If this case happens, $y=b^{2^{-k}}$ is a root of $T_1(y)=0$. Substituting $y=b^{2^{-k}}$ into \eqref{eq:key-identity} gives \begin{equation}\label{eq-T2}
\left(b^{2^{-k}}+\overline b\right)T_2\left(b^{2^{-k}}\right)^{2^k}=0.
\end{equation}
Note that \(\overline b^{2^k}+b\neq 0\). Otherwise $b+\bm^{2^k}= 0$ implies $b\in\fn\cap\mathbb{F}_{2^{m+k}}$. Since $\gcd\left(2^k+1,\,2^m+1\right)=1$, the integers $m/e$ and $k/e$ have different parity.
On the other hand, $y=b^{2^{-k}}$ together with $y\in\mu_{2^m+1}$ implies
$b\in\mu_{2^m+1}$.
Consequently,
\[
b\in \mu_{2^m+1}\cap\mathbb{F}_{2^e}=\{1\},
\]
so that $b=1$, a contradiction.
Thus $\overline b^{2^k}+b\neq 0$.
It then follows from \eqref{eq-T2} that $T_2(b^{2^{-k}})=0$. Hence
\[
\left\{y\in\mu_{2^m+1}: y^{2^k}+b=0\right\}\subseteq N_2.
\]
Moreover, by \eqref{eq-N1N2}, we obtain
\[
\Phi
= N_1\setminus\left(\left\{y\in\mu_{2^m+1}: y^{2^k}+b=0\right\}\cup N_2\right)
= N_1\setminus N_2.
\]

Combining all cases and using $N_2\subseteq N_1$, we conclude that
\[
\#\Phi=\#(N_1\setminus N_2)=\#N_1-\#N_2.
\]
The desired result then follows. This completes the proof.
\end{proof}

By Lemma \ref{lem:gcd}, the analysis is
split into two cases according to the value of
\(\gcd(2^k-1,2^m+1)\).

\begin{thm}\label{thm.DS2}
Assume that $n=2m$, where $k$ and $m$ are positive integers satisfying $\gcd(2^k+1,2^m+1)=1$. Set
$e=\gcd(k,m)$ and $s=\left(2^k+1\right)^{-1}$.
If $\gcd\left(2^k-1,\,2^m+1\right)=2^e+1$, then the power function $F(x)=x^{s(2^m-1)+1}$ over $\F_{2^n}$ is locally differentially $2^e$-uniform and its differential spectrum is given by
$$\mathbb{DS}_F=\left\{\omega_0=2^{n}-2^{n-e}+2^{m-e}-1,\,\omega_{2^e}=2^{n-e}-2^{m-e},\,\omega_{2^m}=1\right\}.$$
\end{thm}

\begin{proof}
By Lemma \ref{lem.b2}, for any $b\in\F_{2^n}\backslash\{1\}$,
${\rm DDT}_F(1,b)=\#N_1-\#N_2$, where
$N_i=\{y\in\mu_{2^m+1}: T_i(y)=0\}$ for $i=1,2$, with $T_1(y)$ and $T_2(y)$ defined in \eqref{eq.T1} and \eqref{eq.T2}, respectively.
We distinguish two cases according to whether $b+\overline b^{2^k}=0$.

{\bf Case 1}: $b+\bm^{2^k}= 0$. 
In this case, the equation $T_1(y)=0$ is equivalent to
\begin{equation}\label{eq-case1}
\left(\bm^{2^k+1}+1\right) y^{2^{2k}}+ \left(b^{2^k+1}+1\right)y=0.
\end{equation} 
Since \(\gcd\left(2^k-1,2^m+1\right)=2^e+1\),
Lemma \ref{lem:gcd} implies that \(k/e\) is even and \(m/e\) is odd. Hence, \(\gcd\left(2^k+1,2^n-1\right)=1\), and consequently,
\(b^{2^k+1}+1\neq0\) for \(b\neq1\).
Therefore, the equation \eqref{eq-case1} becomes $y^{2^{2k}-1}=1,$
which has $2^e+1$ solutions in $\mu_{2^m+1}$ since $\gcd\left(2^{2k}-1,\,2^m+1\right)=2^e+1$ due to \(m/e\) odd, which implies $\#N_1=2^e+1$.
Furthermore, by Lemma \ref{lem.eq}, the equation $T_2(y)=0$
has $0,\,1,\,2$ or $2^{\gcd(k,\,m)}+1=2^e+1$ roots in $\mu_{2^m+1}$.
Combining this with the fact that ${\rm DDT}_F(1,\,b)$ is even, we deduce that
$${\rm DDT}_F(1,\,b)=\#N_1-\#N_2\in\{0,\,2^e\}.$$

{\bf Case 2}: $b+\bm^{2^k}\ne 0$. If this case happens, by Lemma \ref{lem.eq}, the equation
$T_1(y)=0$ has $0,\,1,\,2$ or $2^{\gcd(2k,\,m)}+1=2^{e}+1$ roots in $\mu_{2^m+1}$; while the equation $T_2(y)=0$
has $0,\,1,\,2$ or $2^{\gcd(k,\,m)}+1=2^e+1$ roots in $\mu_{2^m+1}$.
Therefore $$\#N_1-\#N_2\in\left\{0,\,1,\,2,\,
2^{e}-1,\,2^{e},\,2^{e}+1\right\}.$$
Since ${\rm DDT}_F(1,\,b)$ is even, it follows that
$${\rm DDT}_F(1,\,b)\in\{0,\,2,\,2^e\}.$$

Combining these two cases, we conclude that ${\rm DDT}_F(1,\,b)\in\{0,\,2,\,2^e\}$ for $b\in\F_{2^n}\backslash\{1\}$. Substituting $t_1=2$ and $t_2=2^e$ into Theorem \ref{thm.diff-4}, we obtain that
$$\mathbb{DS}_F=\{\omega_0,\,\omega_{2},\,\omega_{2^e},\,\omega_{2^m}=1\}$$
and
\[\omega_0=2^n-2^{n-e}+2^{m-e}-1,\,\omega_{2}=0,\,\omega_{2^e}=2^{n-e}-2^{m-e}.\]
Since $\omega_{2^e}=2^{m-e}(2^m-1)>1$, it follows that $F(x)$ is locally differentially $2^e$-uniform. This completes the proof.
\end{proof}

\begin{remark}
In the special case $e=2$, Theorem \ref{thm.DS2} yields an infinite
family of locally differentially $4$-uniform Niho type power functions.
Furthermore, computational experiments for $1\leq m\leq 10$
show that, apart from the previously known power functions and their
equivalent classes described in Lemma \ref{lem-eqv}, all observed locally differentially $4$-uniform Niho type power functions are contained in the family
constructed in Theorem \ref{thm.DS2}.
\end{remark}

\begin{example}\label{ex1}
Let $n=2m$, $1\leq m\leq 10$, and let $k$ satisfy $\gcd\left(2^k+1, 2^m+1\right)=1$, $\gcd(k, m)=2$, and $\gcd\left(2^k-1, 2^m+1\right)=5$. Then $F(x)=x^{\left(2^k+1\right)^{-1}(2^m-1)+1}$ is locally differentially $4$-uniform and its differential spectrum is given by Table \ref{tab-Diff1}, which is consistent with Theorem \ref{thm.DS2}.
\begin{table}[ht]
\caption{The differential spectrum of $F(x)$ in Example \ref{ex1}}\label{tab-Diff1}
\label{tab-Diff1}
\centering
\begin{tabular}{lll}
\hline
$n$ & $k$ & $\mathbb{DS}_F$\\
\hline
12 &
$4,8$ &
$\{\omega_0=3087,\ \omega_4=1008,\ \omega_{64}=1\}$\\
20 &
$4,8,12,16$ &
$\{\omega_0=786687,\ \omega_4=261888,\ \omega_{1024}=1\}$\\
\hline
\end{tabular}
\end{table}
\end{example}

\begin{thm}\label{thm.DS1}
Assume that $n=2m$, where $k$ and $m$ are positive integers satisfying $\gcd(2^k+1,2^m+1)=1$. Set
$e=\gcd(k,m)$ and $s=\left(2^k+1\right)^{-1}$. If $\gcd\left(2^k-1,\,2^m+1\right)=1$, then the power function $F(x)=x^{s(2^m-1)+1}$ over $\F_{2^n}$ is locally differentially $(2^{2e}-2^e)$-uniform and its differential spectrum is given by
$$\mathbb{DS}_F=\left\{\omega_0,
\,\omega_2,
\,\omega_{2^{2e}-2^e},\,\omega_{2^m}=1\right\},$$
where $\omega_0=2^n-1-\frac{2^{m-e}\left(2^{2e-1}-2^{e}+1\right)(2^m-1)}
{2^{e}-1}$, $\omega_2=\frac{2^{m+e-1}(2^{m}-1)}
{2^{e}+1}$ and $\omega_{2^{2e}-2^e}=\frac {2^{m-e}(2^{m}-1)}
{2^{2e}-1}$.
%\[
%\begin{aligned}
%&\omega_0=\frac{(2^{m+e-1}-2^e-2^{m-e}-1)(2^{m}-1)}
%{2^{e}-1},\,
%\omega_2=\frac{2^{m+e-1}(2^{m}-1)}
%{2^{e}+1},\\
%&\omega_{2^{2e}-2^e}=\frac {2^{m-e}(2^{m}-1)}
%{2^{2e}-1},\quad
%\omega_{2^m}=1.
%\end{aligned}
%\]
\end{thm}

\begin{proof}
By Lemma \ref{lem.b2}, for $b\in\F_{2^n}\backslash\{1\}$, we have
${\rm DDT}_F(1,b)=\#N_1-\#N_2$, where
$N_i=\{y\in\mu_{2^m+1}:T_i(y)=0\}$ for $i=1,2$.
Since $\gcd\left(2^k+1,2^m+1\right)=\gcd\left(2^k-1,2^m+1\right)=1$,
Lemma \ref{lem:gcd} implies that $k/e$ is odd and $m/e$ is even.
Moreover, $\gcd\left(2^{2k}-1,2^m+1\right)=1$.
We distinguish two cases according to whether
$b+\overline b^{2^k}=0$.

{\bf Case 1}: $b+\bm^{2^k}= 0$. In this case, $b\in\mathbb{F}_{2^e}$ due to $\gcd(m+k,n)=e$. Hence $\overline b=b$ and $b^{2^k}=b$. It follows that
\[
T_1(y)=(b^2+1)y^{2^{2k}}+(b^2+1)y.
\]
Since $b\neq1$, the equation $T_1(y)=0$ is equivalent to
$y^{2^{2k}-1}=1$. As $\gcd\left(2^{2k}-1,2^m+1\right)=1$, the above equation admits the unique solution
$y=1$ in $\mu_{2^m+1}$, and hence $\#N_1=1$.
Furthermore,
\[
T_2(1)=1+\overline b+b+1=0,
\]
showing that $1\in N_2$. Recalling from the proof of Theorem \ref{thm.DS2} that $N_2\subseteq N_1$, we obtain $\#N_2=1$.
Therefore ${\rm DDT}_F(1,\,b)=\# N_1-\# N_2=0$.

{\bf Case 2}: $b+\bm^{2^k}\ne0$. By Lemma \ref{lem.eq}, together with the fact that $m/e$ is even, the equation $T_1(y)=0$ has
$0$, $1$, $2$, or $2^{\gcd(2k,m)}+1=2^{2e}+1$ roots in
$\mu_{2^m+1}$, whereas
$T_2(y)=0$ has $0$, $1$, $2$, or $2^{\gcd(k,m)}+1=2^e+1$ roots.
Accordingly,
$$\#N_1-\#N_2\in\left\{0,\,1,\,2,\,2^{2e}-2^e,\,
2^{2e}-1,\,2^{2e},\,2^{2e}+1\right\}.$$
Since ${\rm DDT}_F(1,\,b)$ is even, it follows that
$${\rm DDT}_F(1,\,b)\in\left\{0,\,2,\,2^{2e}-2^e,\,2^{2e}\right\}.$$
We next show that ${\rm DDT}_F(1,\,b)\ne2^{2e}$. Suppose to the contrary that ${\rm DDT}_F(1,\,b)=2^{2e}$. Then $T_1(y)=0$ has exactly $2^{2e}+1$ roots in
$\mu_{2^m+1}$, while
$T_2(y)=0$ has a unique root, denoted by $y_0$.
Since $N_2\subseteq N_1$, we have $T_1(y_0)=0$. In this case, $b\bm\ne1$. Otherwise if $b\bm=1$, then
\[T_2(y)=\left(y+\overline b\right)\big(y^{2^k}+b\big),\]
and hence $T_2(y)=0$ has two distinct solutions
$y=\overline b$ and $y=b^{2^{-k}}$ in $\mu_{2^m+1}$ since
$b+\overline b^{2^k}\neq0$. This contradicts the uniqueness of $y_0$.
Therefore $b\overline b\neq 1$, and consequently $y^{2^k}+b\neq 0$ and $y+\overline b\neq 0$ for all $y\in\mu_{2^m+1}$.
Moreover, $T_2(y_0)=0$ indicates that
\begin{equation}\label{eq1-y0}
y_0=\frac{\overline b y_0^{2^k}+1}{y_0^{2^k}+b},\,\,y_0^{2^k}=\frac{ b y_0+1}{y_0+\bm}.
\end{equation}
Let $y_1\neq y_0$ be another root of $T_1(y)=0$, and define
$y_2=\frac{\overline b y_1^{2^k}+1}{y_1^{2^k}+b}$. Then
\begin{equation}\label{eq1-y2}
y_1^{2^k}=\frac{by_2+1}{y_2+\overline b},
\end{equation}
while a direct substitution into $T_1(y)$ yields
\[T_1(y_2)=\frac{\left(1+b\bm\right)T_1(y_1)^{2^k}}
{\big(y_1^{2^k}+b\big)^{2^{2k}+1}}=0.\]
Hence $y_2$ is also a root of $T_1(y)=0$. We now verify that $y_0$, $y_1$, and $y_2$ are pairwise distinct. Suppose that $y_1=y_2$. Then
$y_1=\big(\overline b y_1^{2^k}+1\big)/\big(y_1^{2^k}+b\big)$, which is equivalent to
$y_1^{2^k+1}+\overline b y_1^{2^k}+by_1+1=0$.
Thus, $T_2(y_1)=0$, contradicting the assumption that
$y_0$ is the unique root of $T_2(y)=0$. Next, suppose that $y_2=y_0$. Combining \eqref{eq1-y0} and \eqref{eq1-y2}, we obtain
\[
y_1^{2^k}
=\frac{by_2+1}{y_2+\overline b}
=\frac{by_0+1}{y_0+\overline b}
=y_0^{2^k},
\]
which forces $y_1=y_0$, a contradiction. Consequently, $y_0$, $y_1$, and $y_2$ are three distinct solutions of $T_1(y)=0$ in $\mu_{2^m+1}$. Invoking the root parametrization established in Lemma \ref{lem.eq}, every root of \(T_1(y)=0\) in
\(\mu_{2^m+1}\setminus\{y_0,y_1,y_2\}\) can be expressed as
\begin{equation}\label{y3}
y_{\xi}=
\frac{y_1y_2+\xi y_0y_2+(1+\xi)y_0y_1}
{y_0+\xi y_1+(1+\xi)y_2}
\end{equation}
for $\xi\in\mathbb{F}_{2^{2e}}\setminus\mathbb{F}_2$. By Lemma \ref{lem.eq-tri}, there exists
\(\xi_0\in\mathbb{F}_{2^{2e}}\setminus\mathbb{F}_2\) satisfying
$\xi_0^{2^e}+\xi_0+1=0$. Substituting $y_2=\big(\bm y_1^{2^k}+ 1\big)/\big(y_1^{2^k}+b\big)$ into \eqref{y3} yields
\begin{equation}\label{y3-y0y1}
y_{\xi_0}=\frac{\beta_1 y_1^{2^k+1}+\beta_2 y_1^{2^k}+\beta_3 y_1+\beta_4}{\alpha_1 y_1^{2^k+1}+\alpha_2 y_1^{2^k}+\alpha_3 y_1+\alpha_4},
\end{equation}
where
\begin{eqnarray}
 && \alpha_1=\xi_0,\,\,\alpha_2=y_0+\bm (\xi_0+1),\,\,\alpha_3=b\xi_0,\,\,\alpha_4=b y_0+\xi_0+1, \label{alpha}\\
&& \beta_1=(\xi_0+1) y_0+\bm ,\,\,\beta_2=\bm \xi_0 y_0,\,\,\beta_3=b(\xi_0+1) y_0+1,\,\,\beta_4=\xi_0 y_0. \label{beta}
\end{eqnarray}
We proceed to show that
\(T_2(y_{\xi_0})=0\), thereby contradicting the assumption that
\(T_2(y)=0\) admits the unique solution \(y_0\) in
\(\mu_{2^m+1}\). For convenience, let $\boldsymbol{\alpha}
=(\alpha_1,\alpha_2,\alpha_3,\alpha_4)$, $
\boldsymbol{\beta}=(\beta_1,\beta_2,\beta_3,\beta_4)$ and
$\phi(\boldsymbol{x})=x_1 y_1^{2^k+1}+x_2 y_1^{2^k}+x_3 y_1+x_4$ with $\boldsymbol{x}=(x_1,x_2,x_3,x_4)\in\F_{2^n}^4$. Then \eqref{y3-y0y1} can be rewritten as $y_{\xi_0}=\frac{\phi(\boldsymbol{\beta})}{\phi(\boldsymbol{\alpha})}$. Substituting it into $T_2(y)$ defined in \eqref{eq.T2} gives
\begin{equation}\label{eq.T2-yxi}
T_2(y_{\xi_0})=\frac{\phi\left(\overline b\boldsymbol{\alpha}+\boldsymbol{\beta}\right)\phi(\boldsymbol{\beta})^{2^k}+\phi(\boldsymbol{\alpha})^{2^k}\phi\left(\boldsymbol{\alpha}+b\boldsymbol{\beta}\right)}{\phi(\boldsymbol{\alpha})^{2^k+1}}. 
\end{equation}
Moreover, we have
$\phi(\boldsymbol{x})^{2^k}=x_1^{2^k} y_1^{2^{2k}+2^k}+x_2^{2^k} y_1^{2^{2k}}+x_3^{2^k} y_1^{2^k}+x_4^{2^k}.$
Since \(y_1\) satisfies \(T_1(y)=0\), it follows that $$\left(\big(b+\bm^{2^k}\big)y_1+\bm^{2^k+1}+1\right)y_1^{2^{2k}}
=\left(b^{2^k+1}+1\right)y_1+b^{2^k}+\bm.$$
We now claim that
$\big(\bm^{2^k}+b\big)y_1+\bm^{2^k+1}+1\ne0$. Otherwise, $y_1=\big(\bm^{2^k+1}+1\big)/\big(\bm^{2^k}+b\big)$ since $\bm^{2^k}+b\ne0$. Using the condition $y_1\in\mu_{2^m+1}$, we obtain
$$\left(\bm^{2^k+1}+1\right)^{2^m+1}+\left(\bm^{2^k}+b\right)^{2^m+1}=(b\bm)^{2^k+1}+(b\bm)^{2^k}+b\bm+1=\left(b\bm+1\right)^{2^k+1}=0,$$
i.e., $b\bm=1$, a contradiction. Thus $\big(\bm^{2^k}+b\big)y_1+\bm^{2^k+1}+1\ne0$ and then by $T_1(y_1)=0$ we deduce
$$y_1^{2^{2k}}=\frac{\big(b^{2^k+1}+1\big)y_1+b^{2^k}+\bm}{\big(b+\bm^{2^k}\big)y_1+\bm^{2^k+1}+1}.$$
Plugging this expression back into $\phi(\boldsymbol{x})^{2^k}$ leads to
\begin{equation}\label{eq.phi2k}
\phi(\boldsymbol{x})^{2^k}=\frac{\varphi_1(x_1,\,x_3)y_1^{2^k+1}+\varphi_2(x_1,\,x_3)y_1^{2^k}+\varphi_1(x_2,\,x_4)y_1+\varphi_2(x_2,\,x_4)}{\big(b+\bm^{2^k}\big)y_1+\bm^{2^k+1}+1},
\end{equation}
where $\varphi_1(z_1,\,z_2)=\big(b^{2^k+1}+1\big)z_1^{2^k}+\big(b+\bm^{2^k}\big)z_2^{2^k}$ and $\varphi_2(z_1,\,z_2)=\big(b^{2^k}+\bm\big)z_1^{2^k}+\big(\bm^{2^k+1}+1\big)z_2^{2^k}$. 
Using the values given in \eqref{alpha} and \eqref{beta}, together with $y_0^{2^k}=\left(by_0+1\right)/\left(y_0+\bm\right)$ and $\xi_0^{2^k}=\xi_0+1$, a direct calculation gives
\[
\begin{aligned}
\varphi_1(\alpha_1,\alpha_3)
&=\left(b^{2^k+1}+1\right)\xi_0^{2^k}
 +\left(b+\bm^{2^k}\right)(b\xi_0)^{2^k} \\
&=\left(\xi_0+1\right)\left(1+b\overline b\right)^{2^k}
=\lambda\left(\overline b\alpha_1+\beta_1\right),
\\[1ex]
\varphi_1(\beta_1,\beta_3)
&=\left(b^{2^k+1}+1\right)\left(\left(\xi_0+1\right)y_0+\bm\right)^{2^k}+\left(b+\bm^{2^k}\right)\left(b(\xi_0+1)y_0+1\right)^{2^k}\\
&=((b\bm)^{2^k}+1)
\left(\left(\xi_0+1\right)^{2^k}y_0^{2^k}+b\right)\\
&=\lambda\left(b(\xi_0+1)y_0+b\bm+\xi_0\right)
=\lambda\left(\alpha_1+b\beta_1\right),
\end{aligned}
\]
where $\lambda=\frac{\left(b\bm+1\right)^{2^k}}{y_0+\bm}$. Here, the relation $\xi_0^{2^k}=\xi_0+1$ follows from $\xi_0^{2^e}=\xi_0+1$ and the fact that $k/e$ is odd.
Applying the same arguments, we further obtain
\[
\begin{aligned}
\varphi_2(\alpha_1,\alpha_3)
 &=\lambda\left(\overline b\alpha_2+\beta_2\right),&
\varphi_1(\alpha_2,\alpha_4)
 &=\lambda\left(\overline b\alpha_3+\beta_3\right),&
\varphi_2(\alpha_2,\alpha_4)
 &=\lambda\left(\overline b\alpha_4+\beta_4\right),\\
\varphi_2(\beta_1,\beta_3)
 &=\lambda\left(\alpha_2+b\beta_2\right),&
\varphi_1(\beta_2,\beta_4)
 &=\lambda\left(\alpha_3+b\beta_3\right),&
\varphi_2(\beta_2,\beta_4)
 &=\lambda\left(\alpha_4+b\beta_4\right).
\end{aligned}
\]
Using the above relations along with \eqref{eq.phi2k}, we arrive at
\[\phi(\boldsymbol{\alpha})^{2^k}=\frac{\lambda\phi\left(\overline b\boldsymbol{\alpha}+\boldsymbol{\beta}\right)}{\big(b+\bm^{2^k}\big)y_1+\bm^{2^k+1}+1},\quad \phi(\boldsymbol{\beta})^{2^k}=\frac{\lambda\phi\left(\boldsymbol{\alpha}+b\boldsymbol{\beta}\right)}{\big(b+\bm^{2^k}\big)y_1+\bm^{2^k+1}+1},\]
which immediately implies $T_2(y_{\xi_0})=0$ by virtue of \eqref{eq.T2-yxi}.

Combining these two cases, we conclude that for $b\in\mathbb{F}_{2^{2m}}\setminus\{1\}$,
$${\rm DDT}_F(1,\,b)\in\{0,\,2,\,2^{2e}-2^e\}.$$
Substituting $t_1=2$ and $t_2=2^{2e}-2^e$ into Theorem \ref{thm.diff-4}, we obtain that
$$\mathbb{DS}_F=\{\omega_0,\,\omega_{2},\,\omega_{2^{2e}-2^e},\,\omega_{2^m}=1\}$$
and
\[\omega_0=2^n-1-\frac{2^{m-e}\left(2^{2e-1}-2^{e}+1\right)(2^m-1)}
{2^{e}-1},\,\omega_2=\frac{2^{m+e-1}(2^{m}-1)}
{2^{e}+1},\,\omega_{2^{2e}-2^e}=\frac {2^{m-e}(2^{m}-1)}
{2^{2e}-1}.\]
Since $\omega_{2^{2e}-2^e}=\frac {2^{m-e}(2^{m}-1)}
{2^{2e}-1}>1$, it follows that $F(x)$ is locally differentially $(2^{2e}-2^e)$-uniform. This completes the proof.

% $\phi(x_i)^{2^k}=x_1^{2^k} y_1^{2^{2k}+2^k}+x_2^{2^k} y_1^{2^{2k}}+x_3^{2^k} y_1^{2^k}+x_4^{2^k}
% =\frac{(x_1^{2^k}(b^{2^k+1}+1)+x_3^{2^k}(b+\bm^{2^k}))y_1^{2^k+1}+(x_1^{2^k}(b^{2^k}+\bm)+x_3^{2^k}(\bm^{2^k+1}+1))y_1^{2^k}+(x_2^{2^k} (b^{2^k+1}+1)+x_4^{2^k}(b+\bm^{2^k}))y_1+(x_2^{2^k}(b^{2^k}+\bm)+x_4^{2^k}(\bm^{2^k+1}+1))}{(b+\bm^{2^k})y_1+\bm^{2^k+1}+1}$
\end{proof}

\begin{remark}
When $e=1$, Theorem \ref{thm.DS1} recovers the differential spectrum of the locally APN case, which coincides with the known result in
\cite{Xie-locally}.
Moreover, the conditions in Theorem \ref{thm.DS1} imply that $m/e$ is even by Lemma \ref{lem:gcd}. Taking $k=m/2$, we obtain
$e=m/2$ and $d=(2^m-1)(2^{m/2}+1)^{-1}+1$.
In this case, the function $F(x)=x^d$ is affine equivalent to $G(x)=x^{2^{3m/2}+2^m+2^{m/2}-1}$, and Theorem \ref{thm.DS1} also
gives the differential spectrum of $G(x)$ previously studied in \cite{Tu2023}. Therefore, Theorem \ref{thm.DS1} provides a unified characterization of these known cases.
\end{remark}

% \begin{remark}
% Note that when $e=1$ in Theorem \ref{thm.DS1}, $f(x)$ is locally differentially $2$-uniform and its differential spectrum  is $\mathbb{DS}_F=\{\omega_0=2^{n-1}+2^{m-1}-1,\,\omega_2=2^{m-1}(2^{m}-1),\,\omega_{2^m}=1\}$, which coins with the result given by Xie et al. \cite{Xie-locally}.  
% The condition $\gcd(2^k+1,2^m+1)=\gcd(2^k-1,2^m+1)=1$ in Theorem \ref{thm.DS1} implies that $m/e$ is even by Lemma \ref{lem:gcd}. Set $k=m/2$, then $e=m/2$ and $\gcd(2^k+1,2^m+1)=\gcd(2^k-1,2^m+1)=1$. then $f(x)$ is locally differentially $2^{m}-2^{m/2}$-uniform and $\mathbb{DS}_F=\{\omega_0= (2^{3m/2-1}-1)(2^{m/2}+1),
% \,\omega_2=2^{3m/2-1}(2^{m/2}-1),
% \,\omega_{2^{m}-2^{m/2}}=2^{m/2},\,\omega_{2^m}=1\}$, which is consistent with the result given by Tu et al. \cite{Tu2023}. That is to say our result contain both Xie et al. and Tu et al.'s results.
% \end{remark}

\begin{example}\label{ex2}
Let $n=2m$, $1\leq m\leq 10$, and let $k$ satisfy $\gcd(2^k+1, 2^m+1)=\gcd(2^k-1, 2^m+1)=1$ and $\gcd(k, m)=2$. Then $F(x)=x^{(2^k+1)^{-1}(2^m-1)+1}$ is locally differentially $12$-uniform and its differential spectrum is given by Table \ref{tab-Diff2}, which is consistent with Theorem \ref{thm.DS1}.
\begin{table}[ht]
\caption{The differential spectrum of $F(x)$ in Example \ref{ex2}}\label{tab-Diff2}
\label{tab-Diff2}
\centering
\begin{tabular}{lll}
\hline
$n$ & $k$ & $\mathbb{DS}_F$\\
\hline
8 &
$2,6$ &
$\{\omega_0=155,\ \omega_2=96,\ \omega_{12}=4,\ \omega_{16}=1\}$\\
16 &
$2,6,10,14$ &
$\{\omega_0=38335,\ \omega_2=26112,\ \omega_{12}=1088,\ \omega_{256}=1\}$\\
\hline
\end{tabular}
\end{table}
\end{example}

\section{Conclusion}\label{conc}
In this paper, we studied the differential properties of Niho type power functions of the form 
$F(x)=x^{s(2^m-1)+1}$ over $\mathbb{F}_{2^{2m}}$.
First, based on the Walsh transform, we provided a general characterization of the differential spectrum of $F(x)$ when $\operatorname{DDT}_F(1,b)$, for $b\neq 1$, takes at most three distinct values. The case $s=(2^k+1)^{-1}\pmod{2^m+1}$ with $\gcd(k,m)=e$ was further considered. By analyzing the relationship between the roots of the quadratic polynomials $T_1(y)$ and $T_2(y)$ over the unit circle $\mu_{2^m+1}$, we derived the possible values of ${\rm DDT}_F(1,b)$ and completely determined the differential spectra of $F(x)$. The resulting family exhibits locally differential uniformity of either $2^e$ or $(2^{2e}-2^e)$, depending on the value of $\gcd(2^k-1,2^m+1)$.
In particular, the former case gives rise to new locally differentially $4$-uniform Niho type power functions, whereas the latter includes several known results as special cases.

%\section*{Acknowledgments}

\end{document}